%% file: DLFH_v5.tex
\documentclass[final]{article}

\usepackage[nonatbib]{nips_2017}

\usepackage{nips_2017}

\usepackage[utf8]{inputenc} 
\usepackage[T1]{fontenc}    
\usepackage{hyperref}       
\usepackage{url}            
\usepackage{booktabs}       
\usepackage{amsfonts}       
\usepackage{nicefrac}       
\usepackage{microtype}      
\usepackage{algorithm}
\usepackage{algorithmic}
\usepackage{graphicx}
\usepackage{subfigure}
\usepackage{makecell}
\usepackage{multirow}
\usepackage{amsmath}
\usepackage{amssymb}
\usepackage{amsthm}
\usepackage{caption}

\newcolumntype{C}[1]{>{\centering\let\newline\\\arraybackslash\hspace{0pt}}m{#1}}

\input{defs.tex}
\title{Discrete Latent Factor Model \\ for Cross-Modal Hashing}

\author{
  Qing-Yuan Jiang, Wu-Jun Li\\
  National Key Laboratory for Novel Software Technology\\
  Collaborative Innovation Center of Novel Software Technology and Industrialization\\
  Department of Computer Science and Technology, Nanjing University, P. R. China \\
  \texttt{jiangqy@lamda.nju.edu.cn, liwujun@nju.edu.cn} \\
}

\begin{document}

\maketitle

\begin{abstract}
Due to its storage and retrieval efficiency, cross-modal hashing~(CMH) has been widely used for cross-modal similarity search in multimedia applications. According to the training strategy, existing CMH methods can be mainly divided into two categories: relaxation-based continuous methods and discrete methods. In general, the training of relaxation-based continuous methods is faster than discrete methods, but the accuracy of relaxation-based continuous methods is not satisfactory. On the contrary, the accuracy of discrete methods is typically better than relaxation-based continuous methods, but the training of discrete methods is time-consuming. In this paper, we propose a novel CMH method, called \underline{d}iscrete \underline{l}atent \underline{f}actor model based cross-modal \underline{h}ashing~(DLFH), for cross modal similarity search. DLFH is a discrete method which can directly learn the binary hash codes for CMH. At the same time, the training of DLFH is efficient. Experiments on real datasets show that DLFH can achieve significantly better accuracy than existing methods, and the training time of DLFH is comparable to that of relaxation-based continuous methods which are much faster than existing discrete methods.
\end{abstract}

\section{Introduction}\label{introduction}
Nearest neighbor~(NN) search plays a fundamental role in many areas including machine learning, data mining, information retrieval, computer vision and so on. In many real applications, there is no need to return exact nearest neighbors for every given query and approximate nearest neighbor~(ANN) is enough to achieve satisfactory performance~\cite{DBLP:conf/focs/AndoniI06,DBLP:conf/stoc/AndoniR15}. Because ANN search might be much faster than exact NN search, ANN search has become an active research topic with a wide range of applications especially for large-scale problems~\cite{DBLP:conf/focs/AndoniI06,DBLP:conf/icml/NeyshaburS15,DBLP:conf/nips/Raziperchikolaei16}.

Among existing ANN search methods, hashing methods have attracted much more attention due to their storage and retrieval efficiency in real applications~\cite{DBLP:conf/nips/WeissTF08,DBLP:conf/nips/RaginskyL09,DBLP:conf/nips/KulisD09,DBLP:conf/icml/WangKC10,DBLP:conf/icml/NorouziF11,DBLP:conf/cvpr/GordoP11,DBLP:conf/icml/LiuWKC11,DBLP:conf/nips/0002FS12,DBLP:conf/icml/RastegariCFHD13,DBLP:conf/nips/LiuMKC14,DBLP:conf/icml/Shrivastava014,DBLP:conf/cvpr/ShenSLS15,DBLP:conf/nips/AndoniILRS15,DBLP:conf/icml/NeyshaburS15,DBLP:conf/cvpr/Carreira-Perpinan15,DBLP:conf/nips/Raziperchikolaei16}. The goal of hashing is to embed data points from the original space into a Hamming space~\cite{DBLP:conf/nips/WeissTF08} where the similarity in the original space is preserved. More specifically, in the Hamming space, each data point will be represented as a binary code. Based on the binary hash code representation, the storage cost can be dramatically reduced, and furthermore we can achieve constant or sub-linear search speed which is much faster than the search speed in the original space~\cite{DBLP:conf/icml/LiuWKC11}.

Early hashing methods are mainly proposed for uni-modal data to perform uni-modal similarity search. In recent years, with the explosive growing of multimedia data in real applications, multi-modal similarity search has attracted a lot of attention. For example, given a text query, a multi-modal similarity search system can return the nearest images or videos in the database. To achieve an efficient performance for large-scale problems, multi-modal hashing~(MMH) has been proposed for multi-modal search~\cite{DBLP:conf/icml/RastegariCFHD13,DBLP:conf/cvpr/DingGZ14,DBLP:conf/aaai/ZhangL14,DBLP:conf/cvpr/LinDH015}. Existing MMH methods can be divided into two major categories: multi-source hashing~(MSH)~\cite{DBLP:conf/mm/SongYHSH11,DBLP:conf/sigir/ZhangWS11} and cross-modal hashing~(CMH)~\cite{DBLP:conf/cvpr/BronsteinBMP10,DBLP:conf/icml/RastegariCFHD13,DBLP:conf/cvpr/DingGZ14}. MSH methods aim to learn binary hash codes by utilizing information from multiple modalities for each point. In other words, all these multiple modalities should be observed for all data points including the query points and those in database under MSH settings. Because it's usually difficult to observe all the modalities in many real applications, the application scenarios for MSH methods are limited. Unlike MSH methods, CMH methods usually require only one modality for a query point to perform search in a database with other modalities. The application scenarios for CMH are more flexible than those for MSH. For example, CMH can perform text-to-image or image-to-text retrieval tasks in real applications. Hence, CMH has gained more attention than MSH~\cite{DBLP:conf/nips/ZhenY12,DBLP:conf/icml/RastegariCFHD13,DBLP:conf/aaai/ZhangL14,DBLP:conf/cvpr/DingGZ14}.

There have appeared many CMH methods. Some CMH methods are unsupervised, including canonical correlation analysis-iterative quantization~(CCA-ITQ)~\cite{DBLP:journals/pami/GongLGP13}, collective matrix factorization hashing~(CMFH)~\cite{DBLP:conf/cvpr/DingGZ14} and alternating co-quantization~(ACQ)~\cite{DBLP:conf/iccv/IrieAT15}. Some others are supervised, including cross modality similarity sensitive hashing~(CMSSH)~\cite{DBLP:conf/cvpr/BronsteinBMP10}, cross view hashing~(CVH)~\cite{DBLP:conf/ijcai/KumarU11}, multi-modal latent binary embedding~(MLBE)~\cite{DBLP:conf/kdd/ZhenY12}, co-regularized hashing~(CRH)~\cite{DBLP:conf/nips/ZhenY12}, predictable dual-view hashing~(PDH)~\cite{DBLP:conf/icml/RastegariCFHD13}, semantic correlation maximization~(SCM)~\cite{DBLP:conf/aaai/ZhangL14}, semantics preserving hashing~(SePH)~\cite{DBLP:conf/cvpr/LinDH015}, and supervised matrix factorization hashing~(SMFH)~\cite{DBLP:conf/ijcai/LiuJWH16}.


According to the training strategy, existing CMH methods, including both unsupervised and supervised methods, can be mainly divided into two categories: relaxation-based continuous methods and discrete methods. Hashing is essentially a discrete learning problem. To avoid the difficulty caused by discrete learning, relaxation-based continuous methods try to solve a relaxed continuous problem with some relaxation strategy. Representative continuous methods include CMFH, CMSSH, CVH, SCM and SMFH. Discrete methods try to directly solve the discrete problem without continuous relaxation. Representative discrete methods include ACQ, MLBE, PDH and SePH. In general, the training of relaxation-based continuous methods is faster than discrete methods, but the accuracy of relaxation-based continuous methods is not satisfactory. On the contrary, the accuracy of discrete methods is typically better than relaxation-based continuous methods, but the training of discrete methods is time-consuming.


In this paper, we propose a novel CMH method, called \underline{d}iscrete \underline{l}atent \underline{f}actor model based cross-modal \underline{h}ashing~(DLFH), for cross modal similarity search. The contributions of DLFH are outlined as follows: (1). DLFH is a supervised CMH method, and in DLFH a novel discrete latent factor model is proposed to model the supervised information. (2). DLFH is a discrete method which can directly learn the binary hash codes without continuous relaxation. (3). A novel discrete learning algorithm is proposed for DLFH, which can be proved to be convergent. Furthermore, the implementation of DLFH is simple, although the mathematical derivation is technical. (4). The training~(learning) of DLFH is still efficient although DLFH is a discrete method. (5). Experiments on real datasets show that DLFH can achieve significantly better accuracy than existing methods, including both relaxation-based continuous methods and existing discrete methods. Experimental results also show that the training speed of DLFH is comparable to that of relaxation-based continuous methods, and is much faster than that of existing discrete methods.


\section{Notations and Problem Definition}\label{notation}
We use bold uppercase letter like $\U$ and bold lowercase letter like $\u$ to denote a matrix and a vector, respectively. The element at the position $(i,j)$ of matrix $\U$ is denoted as $U_{ij}$. The $i$th row of matrix $\U$ is denoted as $\U_{i*}$, and the $j$th column of matrix $\U$ is denoted as $\U_{*j}$. Furthermore, we use $\Vert\cdot\Vert_F$ to denote the Frobenius norm of a matrix. $\U^T$ denotes the transpose of matrix $\U$. $\text{sign}(\cdot)$ is an element-wise sign function.

Without loss of generality, we assume there exist only two modalities in the data although our DLFH can be easily adapted to more modalities. We use $\X=[\x_1,\x_2,\dots,\x_n]^T\in\RB^{n\times d_x}$ and  $\Y=[\y_1,\y_2,\dots,\y_n]^T\in\RB^{n\times d_y}$ to denote the feature vectors of the two modalities~(modality $x$ and modality $y$), where $d_x$ and $d_y$ respectively denote the dimensionality of the feature spaces for two modalities and $n$ is the number of training data. In particular, $\x_i$ and $\y_i$ denote the feature vectors of the two modalities for training point $i$, respectively. Without loss of generality, the data are assumed to be zero-centered which means $\sum_{i=1}^n\x_i=\0$ and $\sum_{i=1}^n\y_i=\0$. Here, we assume that both modalities are observed for all \emph{training} points. However, we do not require that both modalities are observed for \emph{query}~(test) points. Hence, the setting is cross-modal. Actually, DLFH can be easily adapted to cases where some training points are with missing modalities, which will be left for future study. In this paper, we focus on supervised CMH which has shown better accuracy that unsupervised CMH~\cite{DBLP:conf/nips/ZhenY12,DBLP:conf/aaai/ZhangL14,DBLP:conf/cvpr/LinDH015}. In supervised CMH, besides the feature vectors $\X$ and $\Y$, we are also given a cross-modal supervised similarity matrix $\S\in\{0,1\}^{n\times n}$. If $S_{ij}=1$, it means that point $\x_i$ and point $\y_j$ are similar. Otherwise $\x_i$ and $\y_j$ are dissimilar. Here, we assume all elements of $\S$ are observed. But our DLFH can also be adapted for cases with missing elements in $\S$. $S_{ij}$ can be manually labeled by users, or constructed from the labels of point $i$ and point $j$. For example, if point $i$ and point $j$ share common labels, $S_{ij}=1$. Otherwise, $S_{ij} = 0$.

We use $\U\in\{-1,+1\}^{n\times c}$ and $\VV\in\{-1,+1\}^{n\times c}$ to respectively denote the binary codes for modality $x$ and modality $y$, where $\U_{i*}$ and $\VV_{i*}$ respectively denote the binary hash codes of two modalities for point $i$ and $c$ is the length of binary code. The goal of supervised CMH is to learn the binary codes $\U$ and $\VV$, which try to preserve the similarity information in $\S$. In other words, if $S_{ij}=1$, the Hamming distance between $\U_{i*}$ and $\VV_{j*}$ should be as small as possible and vice verse. Furthermore, we also need to learn two hash functions $h_x(\x_q)\in\{-1,+1\}^{c}$ and $h_y(\y_q)\in\{-1,+1\}^{c}$ respectively for modality $x$ and modality $y$, which can compute binary hash codes for any new query point~($\x_q$ or $\y_q$) which is unseen in the training set.

\section{Discrete Latent Factor Model based Cross-Modal Hashing}\label{DLFH}

In this section, we introduce the details of DLFH, including model formulation and learning algorithm.

\subsection{Model Formulation}
Given a binary code pair $\{\U_{i*},\VV_{j*}\}$, we define $\Theta_{ij}$ as: $\Theta_{ij}=\frac{\lambda}{c}\U_{i*}\VV_{j*}^T$,
where $c$ is the code length which is pre-specified, and $\lambda > 0$ is a hyper-parameter denoting a scale factor for tuning.

By using a logistic function, we define $A_{ij}$ as: $A_{ij}=\sigma(\Theta_{ij})=\frac{1}{1+e^{-\Theta_{ij}}}.$
Based on $A_{ij}$, we define the likelihood of the cross-modal similarity $\S$ as: $p(\S\vert\U,\VV)=\prod_{i=1}^n\prod_{j=1}^np(S_{ij}\vert\U,\VV)$, where $p(S_{ij}\vert\U,\VV)$ is defined as: $p(S_{ij}\vert\U,\VV)=S_{ij}A_{ij}+(1-S_{ij})(1-A_{ij})$.

Then the log-likelihood of $\U$ and $\VV$ can be derived as: $L=\log p(\S\vert\U,\VV)
=\sum_{i,j=1}^n\big[S_{ij}\Theta_{ij}-\log(1+e^{\Theta_{ij}})\big]+\text{const}$, where $\text{const}$ is a constant independent of $\U$ and $\VV$.

The model of DLFH tries to maximize the log-likelihood of $\U$ and $\VV$. That is, DLFH tries to solve the following problem:
\begin{align}
\max_{\U,\VV\in\{-1,+1\}^{n\times c}}\;L=\log p(&\S\vert\U,\VV)=\sum_{i=1}^n\sum_{j=1}^n\big[S_{ij}\Theta_{ij}-\log(1+e^{\Theta_{ij}})\big]+\text{const},
\label{obj}
\end{align}


We can find that maximizing the objective function in~(\ref{obj}) exactly matches the goal of hashing. More specifically, the learned binary hash codes try to preserve the similarity information in $\S$.

Please note that latent factor hashing~(LFH)~\cite{DBLP:conf/sigir/ZhangZLG14} has adopted latent factor model for hashing. However, DLFH is different from LFH and is novel due to the following reasons. Firstly, DLFH is for cross-modal supervised hashing but LFH is for uni-modal supervised hashing. Secondly, DLFH is a discrete method which directly learns discrete~(binary) codes without relaxation, but LFH is a relaxation-based continuous method which cannot directly learn the discrete codes.

\subsection{Learning Algorithm}
Problem~(\ref{obj}) is a discrete~(binary) learning problem, which is difficult to solve. One possible solution is to relax the discrete problem to a continuous problem by discarding the discrete~(binary) constraints. Similar relaxation strategies have been adopted by many existing relaxation-based continuous methods like CMFH~\cite{DBLP:conf/cvpr/DingGZ14} and SMFH~\cite{DBLP:conf/ijcai/LiuJWH16}. However, this relaxation may cause the solution to be sub-optimal and the search accuracy to be unsatisfactory.

In this paper, we propose a novel method to directly learn the binary codes without continuous relaxation. The two parameters $\U$ and $\VV$ are learned in an alternating way. More specifically, we design an iterative learning algorithm, and in each iteration we learn one parameter with the other parameter fixed.

\subsubsection{Learning $\U$ with $\VV$ Fixed}

We try to learn $\U$ with $\VV$ fixed. Even if $\VV$ is fixed, it is still difficult to optimize~(learn) the whole $\U$ in one time. For example, if we simply flip the signs of each element to learn $\U$, the total time complexity will be $\OM(2^{n\times c})$ which is very high. Here, we adopt a column-wise learning strategy to optimize one column~(corresponds to one bit for all data points) of $\U$ each time with other columns fixed. The time complexity to directly learn one column is $\OM(2^{n})$ which is still high. Here, we adopt a surrogate strategy~\cite{DBLP:journals/jcgs/LangeHY00} to learn each column, which results in a lower time complexity of $\OM({n^2})$. More specifically, to optimize the $k$th column $\U_{*k}$, we construct a lower bound of $L(\U_{*k})$ and then optimize the lower bound, which can get a closed form solution and make learning procedure simple and efficient. Moreover, the lower-bound based learning strategy can guarantee the solution to converge.

The gradient and Hessian of the objective function $L$ with respect to $\U_{*k}$ can be computed as follows~\footnote{Please note that the objective function $L(\cdot)$ is defined on the whole real space although the variables $\U$ and $\VV$ are constrained to be discrete. Hence, we can still compute the gradient and Hessian for any discrete points $\U$ and $\VV$.}:
\begin{align}
\frac{\partial L}{\partial \U_{*k}}=\frac{\lambda}{c}\sum_{j=1}^n(\S_{*j}-\A_{*j})V_{jk} ,\quad
\frac{\partial^2 L}{\partial \U_{*k}\partial \U_{*k}^T}=-\frac{\lambda^2}{c^2}\text{diag}
(a_{1},a_{2},\cdots,a_{n}), \nonumber
\end{align}
where $\A=[A_{ij}]_{i,j=1}^n$, $a_{i} = \sum_{j=1}^n A_{ij}(1-A_{ij})$, and $\text{diag}(a_{1},a_{2},\cdots,a_{n})$ denotes a diagonal matrix with the $i$th diagonal element being $a_{i}$.

Because $0<A_{ij}<1$, we have $0<A_{ij}(1-A_{ij})<\frac{1}{4}$. We define $\H^{(u)}_k$ as: $\H_k^{(u)}=-\frac{\lambda^2}{4c^2}\sum_{j=1}^n\I = -\frac{n\lambda^2}{4c^2}\I$, where $\I$ is an identity matrix. Then we can get: $\frac{\partial^2 L}{\partial \U_{*k}\partial \U_{*k}^T}\succeq\H_k^{(u)}$, where $\B\succeq\C$ indicates that $\B-\C$ is a positive semi-definite matrix.

Then we construct the lower bound of $L(\U_{*k})$ as follows:
\begin{align}
\widetilde L(\U_{*k})=&L(\U_{*k}(t))+[\U_{*k}-\U_{*k}(t)]^T\frac{\partial L}{\partial \U_{*k}}(t)+\frac{1}{2}[\U_{*k}-\U_{*k}(t)]^T\H^{(u)}_k(t)[\U_{*k}-\U_{*k}(t)],\nonumber
\end{align}
where $\U_{*k}(t)$ is the value of $\U_{*k}$ at the $t$th iteration, and $\frac{\partial L}{\partial \U_{*k}}(t)$ is the gradient with respect to $\U_{*k}(t)$, $\H^{(u)}_k(t) = \H^{(u)}_k$.

Then, we learn the column $\U_{*k}$ by solving the following problem:
\begin{align}
\max_{\U_{*k}}\; \widetilde L(\U_{*k})=\U_{*k}^T\Big[\frac{\partial L}{\partial \U_{*k}}(t)-\H^{(u)}_k(t)\U_{*k}(t)\Big]+\text{const},\quad\text{s.t.}\;\U_{*k}\in\{-1,+1\}^{n}.\label{obj:Um}
\end{align}
Thus, we can get a closed form solution for problem~(\ref{obj:Um}), and use this solution to get $\U_{*k}(t+1)$:
\begin{align}\label{eq:UkUpdate}
\U_{*k}(t+1)=\text{sign}[\frac{\partial L}{\partial \U_{*k}}(t)-\H^{(u)}_k(t)\U_{*k}(t)].
\end{align}

\subsubsection{Learning $\VV$ with $\U$ Fixed}

When $\U$ is fixed, we adopt a similar strategy as that for $\U$ to learn $\VV$.

%
%
Specifically, we can get the following closed form solution to update $\VV_{*k}$:
\begin{align}\label{eq:VkUpdate}
\VV_{*k}(t+1)=\text{sign}[\frac{\partial L}{\partial \VV_{*k}}(t)-\H^{(v)}_k(t)\VV_{*k}(t)],
\end{align}
where $\frac{\partial L}{\partial \VV_{*k}} = \frac{\lambda}{c}\sum_{i=1}^n(\S_{i*}^T-\A_{i*}^T)U_{ik}$ and $\H^{(v)}_k(t)= -\frac{n\lambda^2}{4c^2}\I$.

The learning algorithm for DLFH is summarized in Algorithm~\ref{alg:Framwork}, which can be easily implemented.
\begin{algorithm}[t]
\caption{Learning algorithm for DLFH}
\label{alg:Framwork}
\begin{algorithmic}
\REQUIRE
    $\S\in\{1,0\}^{n\times n}$: supervised similarity matrix, $c$: code length.
\ENSURE
    $\U$ and $\VV$: binary codes for two modalities.
\STATE {\bf Procedure}: initialize $\U$ and $\VV$.
\FOR {$t=1\to T$}
    \FOR {$k=1\to c$}
        \STATE Update $\U_{*k}$ according to (\ref{eq:UkUpdate}).
    \ENDFOR
    \FOR {$k=1\to c$}
        \STATE Update $\VV_{*k}$ according to (\ref{eq:VkUpdate}).
    \ENDFOR
\ENDFOR
\end{algorithmic}
\end{algorithm}

\begin{theorem}\label{theorem:converge}
The learning algorithm for DLFH which is shown in Algorithm~\ref{alg:Framwork} is convergent.
\end{theorem}

\begin{proof}
Due to space limitation, we just give a brief proof here.

Each time we update either $\U_{*k}$ or $\VV_{*k}$, the objective function value will not decrease. That is to say, we can guarantee that $L(\U_{*k}(t+1))\geq L(\U_{*k}(t))$ and $L(\VV_{*k}(t+1))\geq L(\VV_{*k}(t))$. Furthermore, the objective function $L(\cdot)$ is upper-bounded by 0. Hence, Algorithm~\ref{alg:Framwork} will converge. Because $L(\cdot)$ is non-convex, the learned solution will converge to a local optimum.
\end{proof}

\subsubsection{Stochastic Learning Strategy}\label{sec:stochaticLearning}
We can find that the computational cost for learning $\U_{*k}$ and $\VV_{*k}$ is $\OM(n^2)$. DLFH will become intractable when the size of training set is large. Here, we design a stochastic learning strategy to avoid high computational cost.

It is easy to see that the high computational cost mainly comes from the gradient computation for both $\U_{*k}$ and $\VV_{*k}$. In our stochastic learning strategy, we randomly sample $m$ columns~(rows) of $\S$ to compute $\frac{\partial L}{\partial \U_{*k}}$~($\frac{\partial L}{\partial \VV_{*k}}$) during each iteration. Then, we get the following formulas for updating $\U_{*k}$ and $\VV_{*k}$:
\begin{align}
\U_{*k}(t+1)=\text{sign}\big[\frac{\lambda}{c}&\sum_{q=1}^m(\S_{*j_q}-\A_{*{j_q}})V_{j_qk}(t) +\frac{m\lambda^2}{4c^2}\U_{*k}(t)\big], \label{opt:Uk} \\
\VV_{*k}(t+1)= \text{sign}\big[\frac{\lambda}{c}&\sum_{q=1}^m(\S_{i_q*}^T-\A_{i_q*}^T)U_{i_qk}(t)+\frac{m\lambda^2}{4c^2}\VV_{*k}(t)\big], \label{opt:Vk}
\end{align}
where $\{j_q\}_{q=1}^m$ and $\{i_q\}_{q=1}^m$ are the $m$ sampled column and row indices, respectively.

To get the stochastic learning algorithm for DLFH, we only need to substitute (\ref{eq:UkUpdate}) and (\ref{eq:VkUpdate}) in Algorithm~\ref{alg:Framwork} by (\ref{opt:Uk}) and (\ref{opt:Vk}), respectively. Then the computational cost will decrease from $\OM(n^2)$ to $\OM(nm)$, where $m$ is typically far less than $n$.

\subsection{Out-of-Sample Extension}\label{outofsmaple}
For any unseen query points $\x_q\notin\X$ or $\y_q\notin\Y$, we learn two linear hash functions $h_{x}(\x)$ and $h_{y}(\y)$. Specifically, we minimize the following square loss:
\begin{align}
L_{x}(\W_{x})=\Vert\U-\X\W_{x}\Vert^2_F+\gamma_{x}\Vert\W_{x}\Vert^2_F,\quad L_{y}(\W_{y})=\Vert\VV-\Y\W_{y}\Vert^2_F+\gamma_{y}\Vert\W_{y}\Vert^2_F,\nonumber
\end{align}
where $\gamma_{x}$ and $\gamma_{y}$ are the hyper-parameters for regularization terms. Then, we can get: $\W_{x}=(\X^T\X+\gamma_{x}\I)^{-1}\X^T\U,\; \W_{y}=(\Y^T\Y+\gamma_{y}\I)^{-1}\Y^T\VV$. Then we can get the hash functions for out-of-sample extention: $h_x(\x_q)=\text{sign}(\W_{x}^T \x_q),\;h_y(\y_q)=\text{sign}(\W_y^T \y_q)$.

Please note that here we just use linear hash function to demonstrate the effectiveness of our method. In fact, we can adopt stronger functions, such as deep neural networks, for out-of-sample extension to get better performance. But this is not the focus of this paper, and will be left for future study.

\subsection{Complexity Analysis}
We call the DLFH version without stochastic learning strategy \emph{full DLFH}, and call the stochastic version of DLFH \emph{stochastic DLFH}. The time complexity of full DLFH is $\OM(Tcn^2)$, where $T$ and $c$ are typically small constants. The time complexity of stochastic DLFH is $\OM(Tcnm)$, which is much lower than that of the full DLFH when $m\ll n$. In practice, we suggest to adopt the stochastic DLFH because we find that the accuracy of stochastic DLFH is comparable to that of the full DLFH even if $m$ is set to be $c$ which is typically a small constant. When $m = c$, the training of stochastic DLFH is very efficient.


\section{Experiments}\label{experiment}
We utilize two real datasets to evaluate DLFH and baselines. The experiments are performed on a workstation with Intel~(R) CPU E5-2620V2@2.1G of 12 cores and 64G RAM.

\subsection{Datasets}

Two datasets, MIRFLICKR-25K~\cite{DBLP:conf/mir/HuiskesL08} and NUS-WIDE~\cite{DBLP:conf/civr/ChuaTHLLZ09}, are used for evaluation.

The MIRFLICKR-25K dataset contains 25,000 data points, each of which corresponds to an image associated with some textual tags. We only select those points which have at least 20 textual tags for our experiment. The text for each point is represented as a 1386-D bag-of-words~(BOW) vector. The image for each point is represented by a 512-D GIST feature vector. Additionally, each point is manually annotated with one of the 24 unique labels.

The NUS-WIDE dataset contains 260,648 data points, each of which corresponds to an image associated with textual tags. Each point is annotated with one or multiple labels from 81 concept labels. As in SCM~\cite{DBLP:conf/aaai/ZhangL14} and SePH~\cite{DBLP:conf/cvpr/LinDH015}, we select 186,577 points that belong to the 10 most frequent concepts from the original dataset in our experiments. The text for each point is represented as a 1000-D BOW vector. Furthermore, each image is a 500-D bag-of-visual words~(BOVW) vector.

For both two datasets, the ground-truth neighbors~(similar pairs) are defined as those image-text pairs which share at least one common label.

\subsection{Baselines and Evaluation Protocol}

Six state-of-the-art CMH methods are adopted as baselines for comparison. They are SMFH~\cite{DBLP:conf/ijcai/LiuJWH16}, SePH~\cite{DBLP:conf/cvpr/LinDH015}, SCM~\cite{DBLP:conf/aaai/ZhangL14}, CMFH~\cite{DBLP:conf/cvpr/DingGZ14}, CCA-ITQ~\cite{DBLP:journals/pami/GongLGP13} and MLBE~\cite{DBLP:conf/kdd/ZhenY12}. Other methods, such as PDH~\cite{DBLP:conf/icml/RastegariCFHD13}, are not adopted for comparison because they have been found to be outperformed by the adopted baselines like SMFH and SePH. Among the six baselines, CCA-ITQ and CMFH are unsupervised, and others are supervised. CMFH, SCM and SMFH are relaxation-based continuous methods, others are discrete methods. SePH is a kernel-based method. Following its authors' suggestion, we utilize RBF kernel and randomly take 500 data points as kernel bases to learn hash functions. For the other baselines, we set the parameters by following suggestions of the corresponding authors. For our DLFH, we set $\lambda=8$ and $T=30$ which is selected based on a validation set. We adopt the stochastic version of DLFH unless otherwise stated. In stochastic DLFH, $m=c$. All experiments are run 5 times to remove randomness, then the average performance is reported.

Besides a linear classifier for out-of-sample extension mentioned in Section~\ref{outofsmaple}, we also utilize a kernel logistic regression classifier, which is the same as that in SePH, for out-of-sample extension in our DLFH to learn binary codes for both modalities. This kernel version of DLFH is called KDLFH.

We randomly sample 2,000 and 1,867 data points from the whole dataset to construct a query~(test) set and the rest as retrieval~(database) set for MIRFLICKR-25K and NUS-WIDE datasets, respectively. For our DLFH and most baselines except SePH and MLBE, we use the whole retrieval set as training set. Because SePH and MLBE can not scale to large training set, we randomly sample 5,000 and 1,000 data points from retrieval set to construct training set for SePH and MLBE, respectively. We do not adopt 5,000 data points for training MLBE due to an out-of-memory error.

The mean average precision~(MAP)~\cite{DBLP:conf/aaai/ZhangL14} is the most popular evaluation metric for hashing and it's adopted to evaluate our DLFH and baselines.
\subsection{Convergence Analysis}
To verify the convergence property of DLFH, we conduct an experiment on a subset of MIRFLICKR-25K dataset. Specifically, we randomly select 5,000 data points from MIRFLICKR-25K, with 3,000 data points for training and the rest for test~(query).

Figure~\ref{fig:conv} shows the convergence of objective function value (subfigure (a), (c)) and MAP (subfigure (b), (d)) with 32 and 64 bits, where ``DLFH-Full'' denotes the full version of DLFH and ``DLFH-Stochastic'' denotes the stochastic version of DLFH. In the figure, ``$I \to T$'' denotes image-to-text retrieval where we use image modality as queries and then retrieve text from database, and other notations are defined similarly. We can find that the objective function value of DLFH-Full will not decrease as iteration number increases, which verifies the claim about the convergence in Theorem~\ref{theorem:converge}. There is some vibration in DLFH-Stochastic due to the stochastic sampling procedure, but the overall trend is convergent. For MAP, we can also observe the convergence for overall trend. Both DLFH-Full and DLFH-Stochastic converge very fast, and only a small number of iterations are needed.

Another interesting phenomenon is that the accuracy~(MAP) of DLFH-Stochastic is almost the same as that of DLFH-Full. Hence, unless otherwise stated, the DLFH in the following experiment refers to the stochastic version of DLFH.
\begin{figure}[t]
\centering\small
\begin{tabular}{c@{}@{}c@{}@{}c@{}@{}c}
\begin{minipage}{0.24\linewidth}\centering
    \includegraphics[width=1\textwidth,height = 0.85\textwidth]{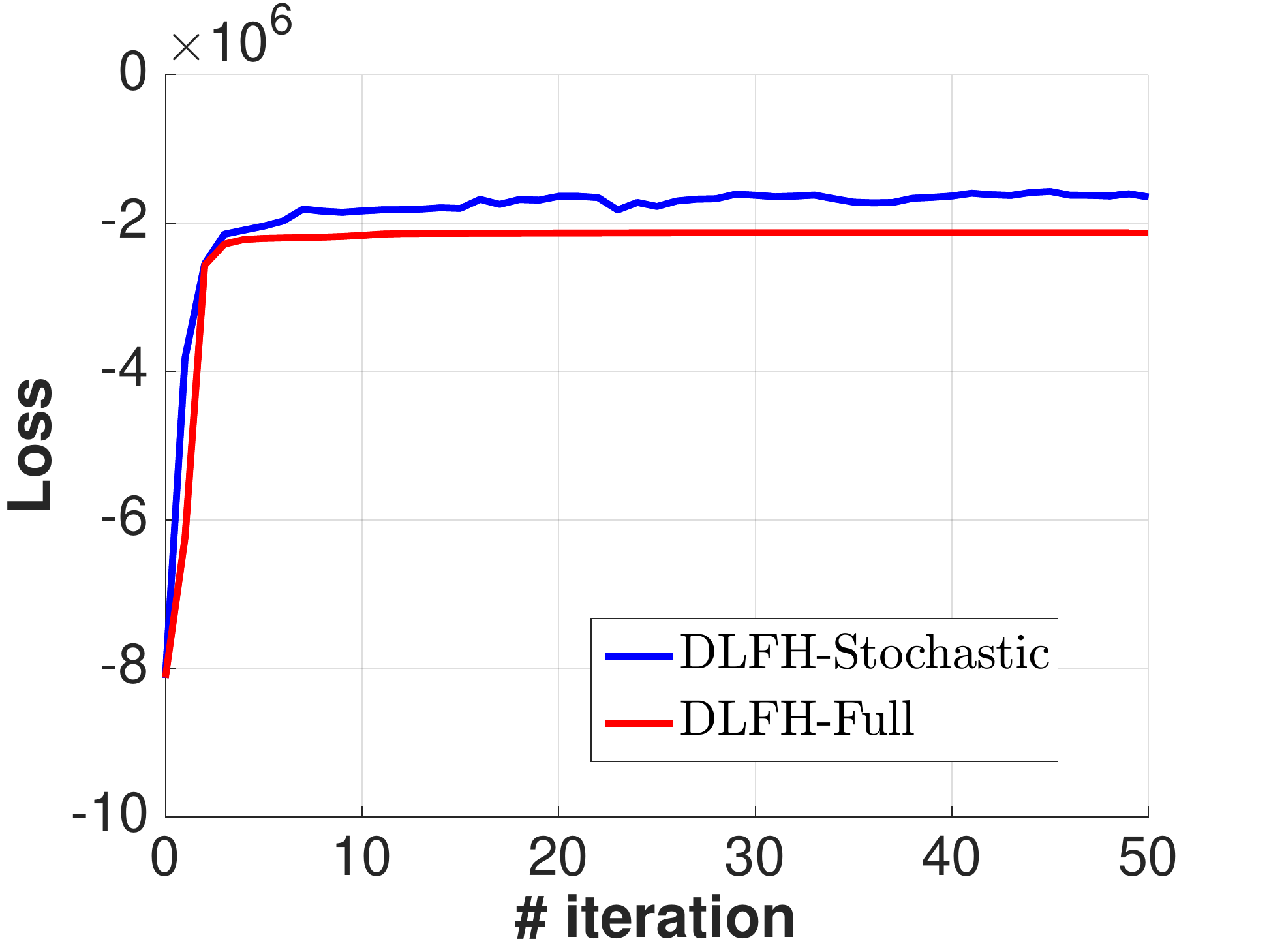}\\
    (a) Obj. value@32 bits
\end{minipage} &
\begin{minipage}{0.24\linewidth}\centering
    \includegraphics[width=1\textwidth,height = 0.85\textwidth]{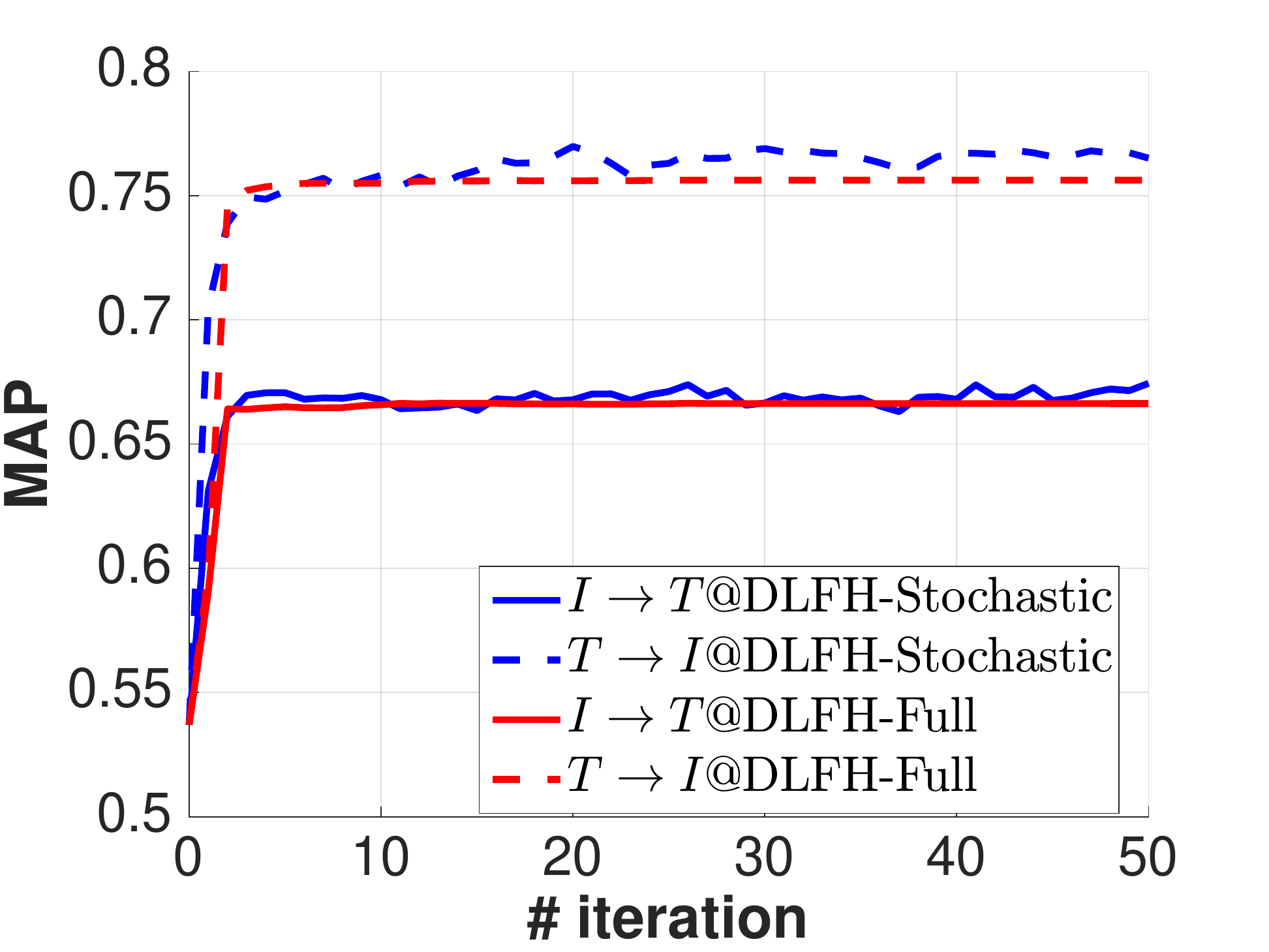}\\
    (b) MAP@32 bits
\end{minipage}
\begin{minipage}{0.24\linewidth}\centering
    \includegraphics[width=1\textwidth,height = 0.85\textwidth]{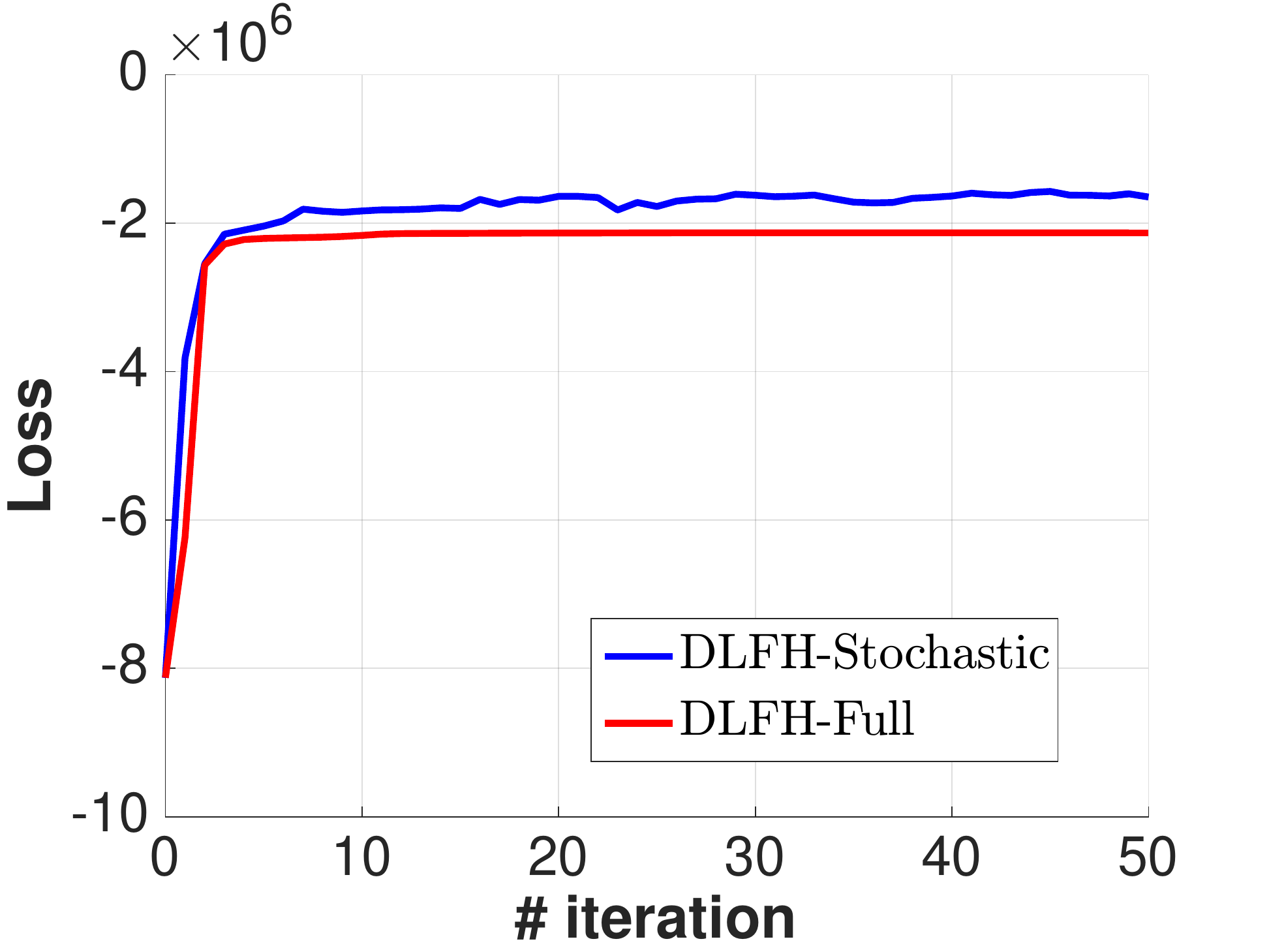}\\
    (c) Obj. value@64 bits
\end{minipage} &
\begin{minipage}{0.24\linewidth}\centering
    \includegraphics[width=1\textwidth,height = 0.85\textwidth]{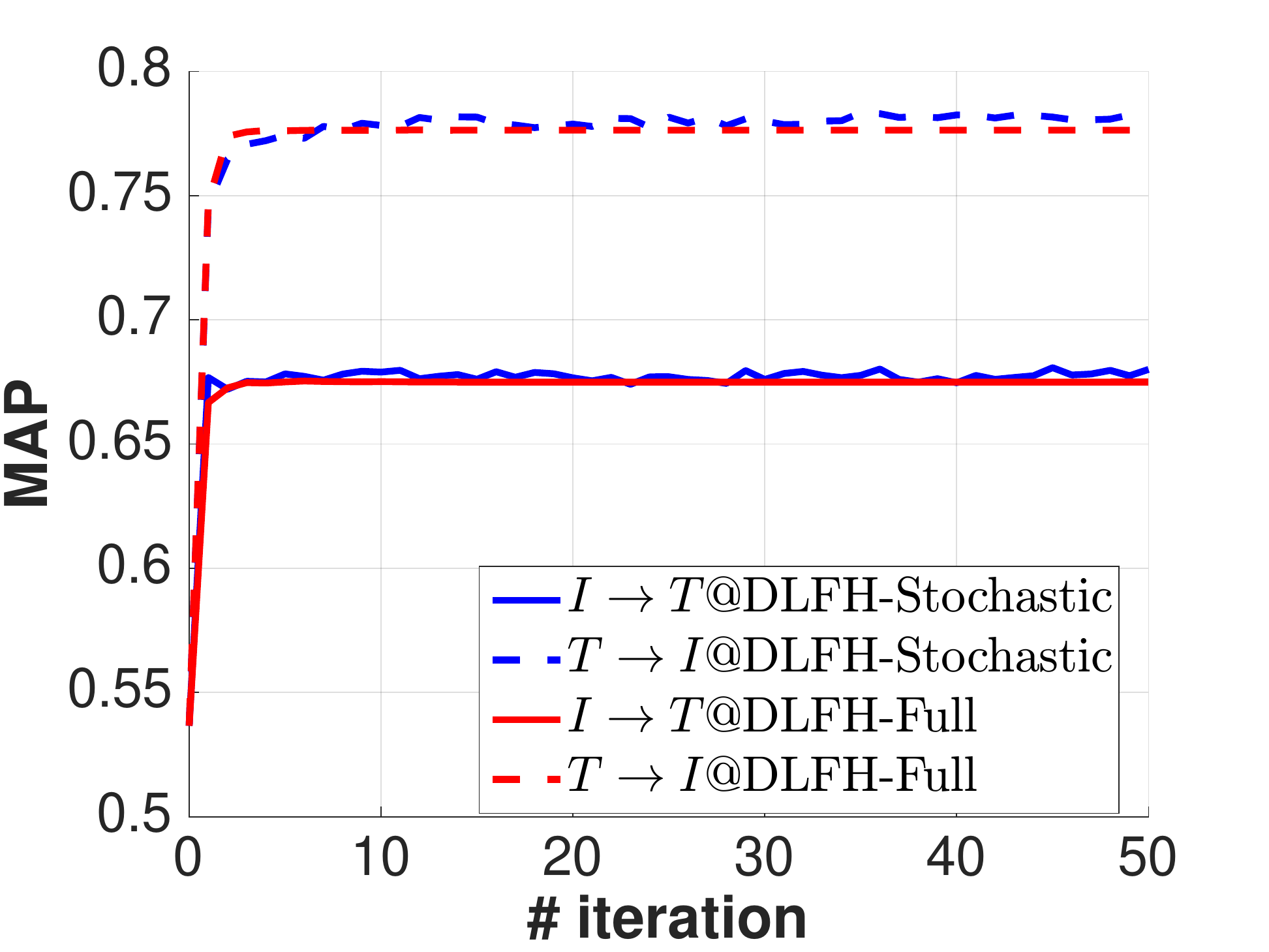}\\
    (d) MAP@64 bits
\end{minipage}
\end{tabular}
\caption{Convergence analysis of DLFH.}
\label{fig:conv}
\end{figure}

\subsection{Accuracy}
Table~\ref{table:map} presents the MAP results on two datasets. By comparing SePH, SMFH, SCM to CMFH and CCA-ITQ, we can see that supervised methods can outperform unsupervised methods in most cases.  By comparing SePH to other methods, we can find that in general discrete methods can outperform relaxation-based continuous methods. Please note that MLBE is a special case because the training of MLBE is too slow and we have to sample only a very small subset for training. Hence, its accuracy is low although it is supervised and discrete. We can find that our DLFH can significantly outperform all baselines in all cases for both image-to-text and text-to-image retrieval tasks. Furthermore, KDLFH can further improve retrieval accuracy thanks to the non-linear out-of-sample extension strategy.

\begin{table*}[t]
\centering\small
\caption{MAP (avg.$\pm$std.). The best accuracy is shown in boldface.}
\label{table:map}
\begin{tabular}{|c|c||c|c|c||c|c|c|}
 \hline
 \multirow{2}{*}{Task} & \multirow{2}{*}{Method} &
 \multicolumn{3}{c||}{{MIRFLICKR-25K}} &  \multicolumn{3}{c|}{{NUS-WIDE}}\\
 \cline{3-8} & & 16 bits & 32 bits & 64 bits & 16 bits & 32 bits & 64 bits \\
 \hline \hline
\cline{2-8}
 & DLFH & {\bf .705$\pm$.002} & {\bf .718$\pm$.005} & {\bf .723$\pm$.002} & {\bf .635$\pm$.005} & {\bf .658$\pm$.003} & {\bf .668$\pm$.003} \\
 \cline{2-8}
 & SMFH & {.584$\pm$.005} & {.608$\pm$.004} & {.624$\pm$.006} & {.447$\pm$.019} & {.457$\pm$.004} & {.463$\pm$.003} \\\cline{2-8}
 & SCM & {.636$\pm$.005} & {.644$\pm$.005} & {.653$\pm$.005} & {.522$\pm$.000} &{.548$\pm$.000} & {.556$\pm$.000}\\\cline{2-8}
\multirow{2}{*}{$I\to T$}
 & CMFH & {.586$\pm$.004} & {.586$\pm$.003} & {.585$\pm$.003} & {.439$\pm$.006} & {.444$\pm$.008} & {.451$\pm$.001}\\\cline{2-8}
 & CCA-ITQ & {.576$\pm$.001} & {.570$\pm$.001} & {.567$\pm$.001} & {.396$\pm$.000} & {.382$\pm$.000} & {.370$\pm$.000} \\\cline{2-8}
 & MLBE & {.560$\pm$.006} & {.548$\pm$.003} & {.543$\pm$.014} & {.334$\pm$.001} & {.336$\pm$.001} & {.361$\pm$.003}\\\Xcline{2-8}{1pt}
 & KDLFH  & {\bf .768$\pm$.004} & {\bf .790$\pm$.004} & {\bf .800$\pm$.001}  & {\bf .684$\pm$.002} & {\bf .700$\pm$.003} & {\bf .715$\pm$.001}\\\cline{2-8}
 & SePH & {.657$\pm$.004} & {.659$\pm$.003} & {.663$\pm$.003} & {.546$\pm$.006} & {.550$\pm$.007} & {.557$\pm$.007}\\
 \cline{2-8}
\hline\hline
 & DLFH & {\bf .781$\pm$.004} & {\bf .803$\pm$.002} & {\bf .817$\pm$.003}  & {\bf .769$\pm$.006} & {\bf .801$\pm$.004} & {\bf .818$\pm$.002}\\\cline{2-8}
 & SMFH & {.569$\pm$.003} & {.584$\pm$.006} & {.597$\pm$.006} & {.413$\pm$.021} & {.411$\pm$.003} & {.402$\pm$.002}\\\cline{2-8}
 & SCM & {.626$\pm$.006} & {.634$\pm$.008} & {.643$\pm$.006}  & {.492$\pm$.000} & {.515$\pm$.000} & {.523$\pm$.000}\\\cline{2-8}
\multirow{2}{*}{$T\to I$}
 & CMFH & {.582$\pm$.005} & {.582$\pm$.003} & {.583$\pm$.003} & {.427$\pm$.006} & {.432$\pm$.005} & {.439$\pm$.002} \\\cline{2-8}
 & CCA-ITQ & {.575$\pm$.001} & {.570$\pm$.001} & {.567$\pm$.001} & {.390$\pm$.000} & {.378$\pm$.000} & {.368$\pm$.000}\\\cline{2-8}
 & MLBE & {.594$\pm$.024} & {.568$\pm$.014} & {.534$\pm$.055} & {.345$\pm$.000} & {.345$\pm$.000} & {.346$\pm$.000} \\\Xcline{2-8}{1pt}
 & KDLFH & {\bf .816$\pm$.006} & {\bf .845$\pm$.004} & {\bf .850$\pm$.004} & {\bf .776$\pm$.014} & {\bf .815$\pm$.005} & {\bf .829$\pm$.003} \\\cline{2-8}
 & SePH & {.689$\pm$.005} & {.695$\pm$.005} & {.698$\pm$.003} & {.632$\pm$.009} & {.642$\pm$.005} & {.650$\pm$.009}\\
 \cline{2-8}
\hline
 \end{tabular}
\end{table*}

%

\subsection{Training Speed}

To evaluate the training speed of DLFH, we adopt different number of data points from retrieval set to construct training set and then report the training time. Table~\ref{table:traintime} presents the training time for our DLFH and baselines. Please note that ``-'' denotes that we cannot carry out corresponding experiments due to out-of-memory errors. We can find that the unsupervised method CCA-ITQ is the fasted method because it does not use supervised information for training. Although the training of CCA-ITQ is fast, the accuracy of it is low. Hence, CCA-ITQ is not practical in real applications. By comparing MLBE and SePH to SMFH, SCM and CMFH, we can find that existing discrete methods are much slower than relaxation-based continuous methods. Although our DLFH is a discrete method, its training speed can be comparable to relaxation-based continuous methods. Furthermore, KDLFH is also much faster than the kernel baseline SePH.

Although KDLFH can achieve better accuracy than DLFH, the training speed of KDLFH is much slower than that of DLFH. Hence, in real applications, we provide users with two choices between DLFH and KDLFH based on whether they care more about training speed or accuracy.

Overall, our DLFH and KDLFH methods achieve the best accuracy with a relatively fast training speed. In particular, our methods can significantly outperform relaxation-based continuous methods in terms of accuracy, but with a comparable training speed. Furthermore, our methods can significantly outperform existing discrete methods in terms of both accuracy and training speed.
\begin{table}[t]
\centering\small
\caption{Training time~(in second) on subsets of NUS-WIDE.}
\label{table:traintime}
\begin{tabular}{|c||c|c|c|c|c|c|}
 \hline
 Method & 1K & 5K & 10K & 50K & 100K & $\sim$184K\\
 \hline
 DLFH & 1.26 & 3.84 & 6.61 & 34.88 & 69.25 & 112.88\\\cline{1-7}
 SMFH & 1.60 & 9.84 & 14.15 & 46.94 & 82.15 & 161.04\\\cline{1-7}
 SCM & 10.41 & 11.13 & 11.09 & 11.34 & 13.53 & 12.30\\\cline{1-7}
 CMFH & 4.20 & 12.01 & 20.65 & 84.00 & 165.54 & 305.51 \\\cline{1-7}
 CCA-ITQ & 0.56 & 0.59 & 0.69 & 1.51 & 2.52 & 4.25 \\\cline{1-7}
 MLBE & 998.98 & - & - & - & - & -\\\Xcline{1-7}{1pt}
 KDLFH & 56.58 & 214.65 & 479.46 & 2097.52 & 4090.21 & 7341.05 \\\cline{1-7}
 SePH & 80.54 & 606.18 & - & - & - & -\\\cline{1-7}
 \hline
 \end{tabular}
\end{table}
\subsection{Sensitivity to Hyper-Parameter}
The most important hyper-parameters for DLFH is $\lambda$ and the number of sampled points $m$. We explore the influence to these two hyper-parameters.

We report the MAP values for different $\lambda$ from the range of $[10^{-4},32]$ on two datasets with the code length being 16 bits. The results are shown in Figure~\ref{fig:hyperparam}~(a) and Figure~\ref{fig:hyperparam}~(b). We can find that DLFH is not sensitive to $\lambda$ in a large range when $1<\lambda<16$.

Furthermore, we present the influence of $m$ in Figure~\ref{fig:hyperparam}~(c) and Figure~\ref{fig:hyperparam}~(d). We can find that higher accuracy can be achieved by sampling more points in each iteration of stochastic DLFH. However, more sampled points will require higher computation cost. Hence, we simply set $m=c$ in our experiment to get a good tradeoff between accuracy and efficiency.

\begin{figure}[t]
\centering
\begin{tabular}{c@{}@{}c@{}@{}c@{}@{}c}
\begin{minipage}{0.24\linewidth}\centering
    \includegraphics[width=1\textwidth,height = 0.85\textwidth]{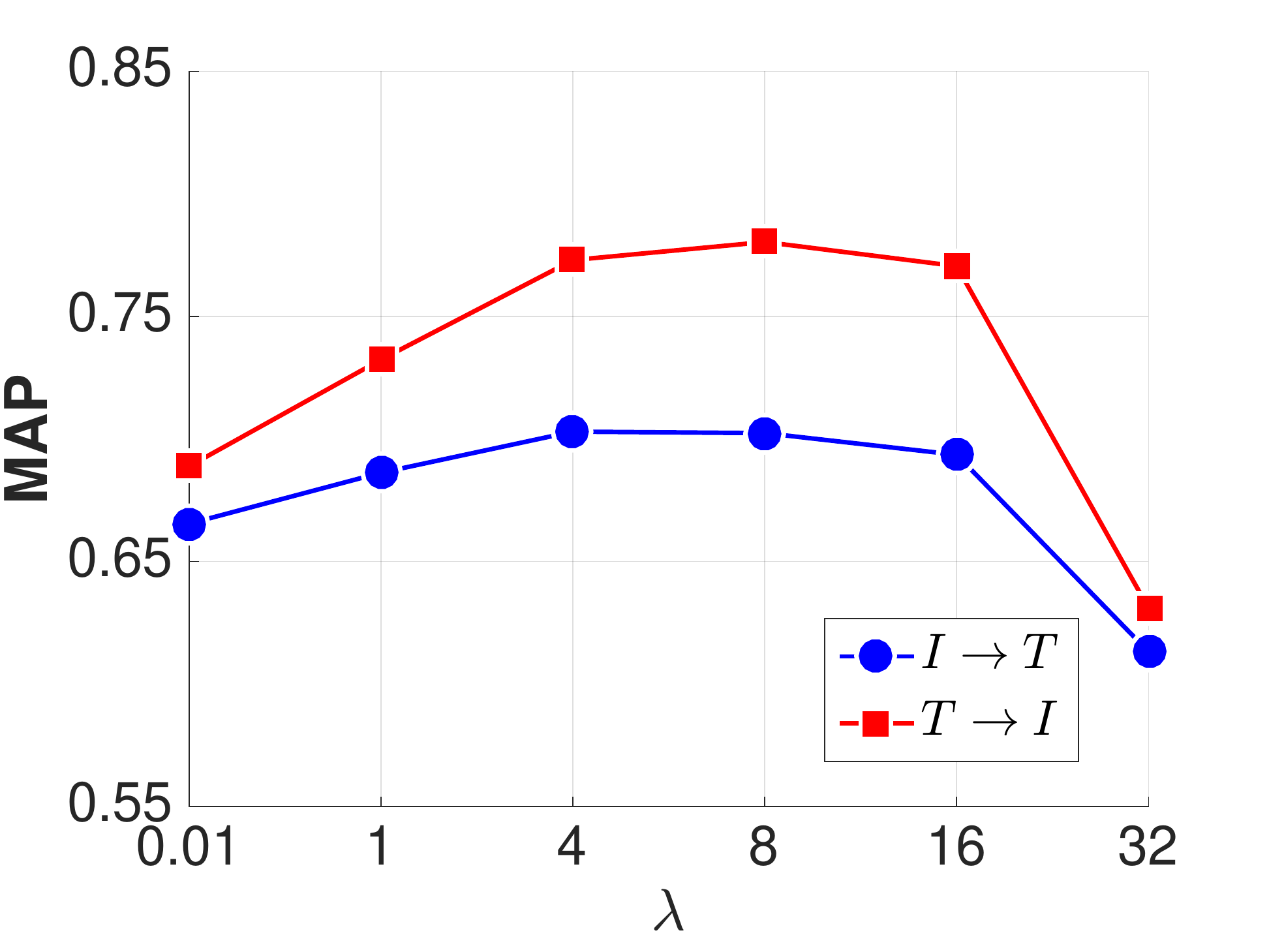}\\
    (a) MIRFLICKR-25K
\end{minipage} &
\begin{minipage}{0.24\linewidth}\centering
    \includegraphics[width=1\textwidth,height = 0.85\textwidth]{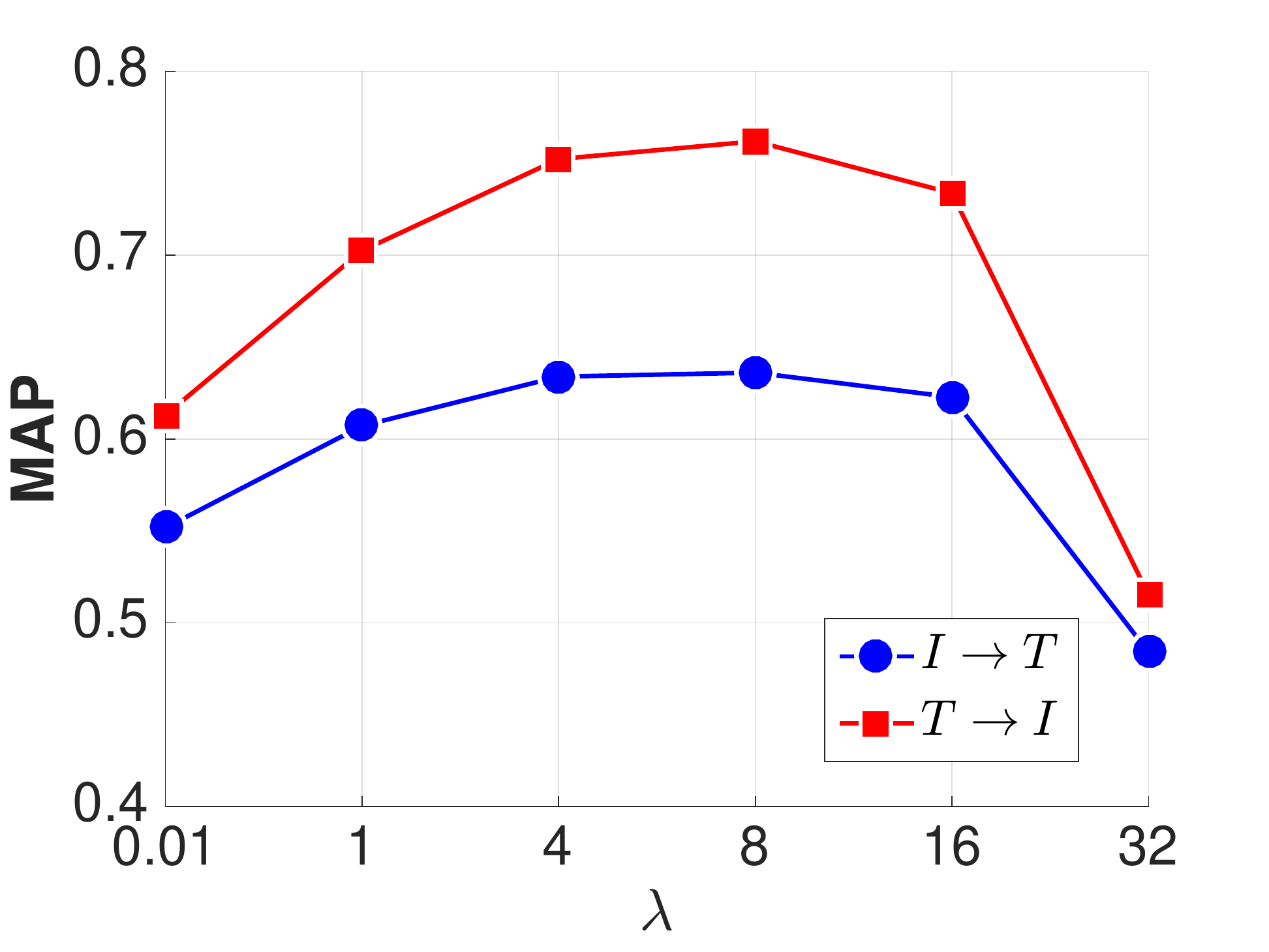}\\
    (b) NUS-WIDE
\end{minipage} &
\begin{minipage}{0.24\linewidth}\centering
    \includegraphics[width=1\textwidth,height = 0.85\textwidth]{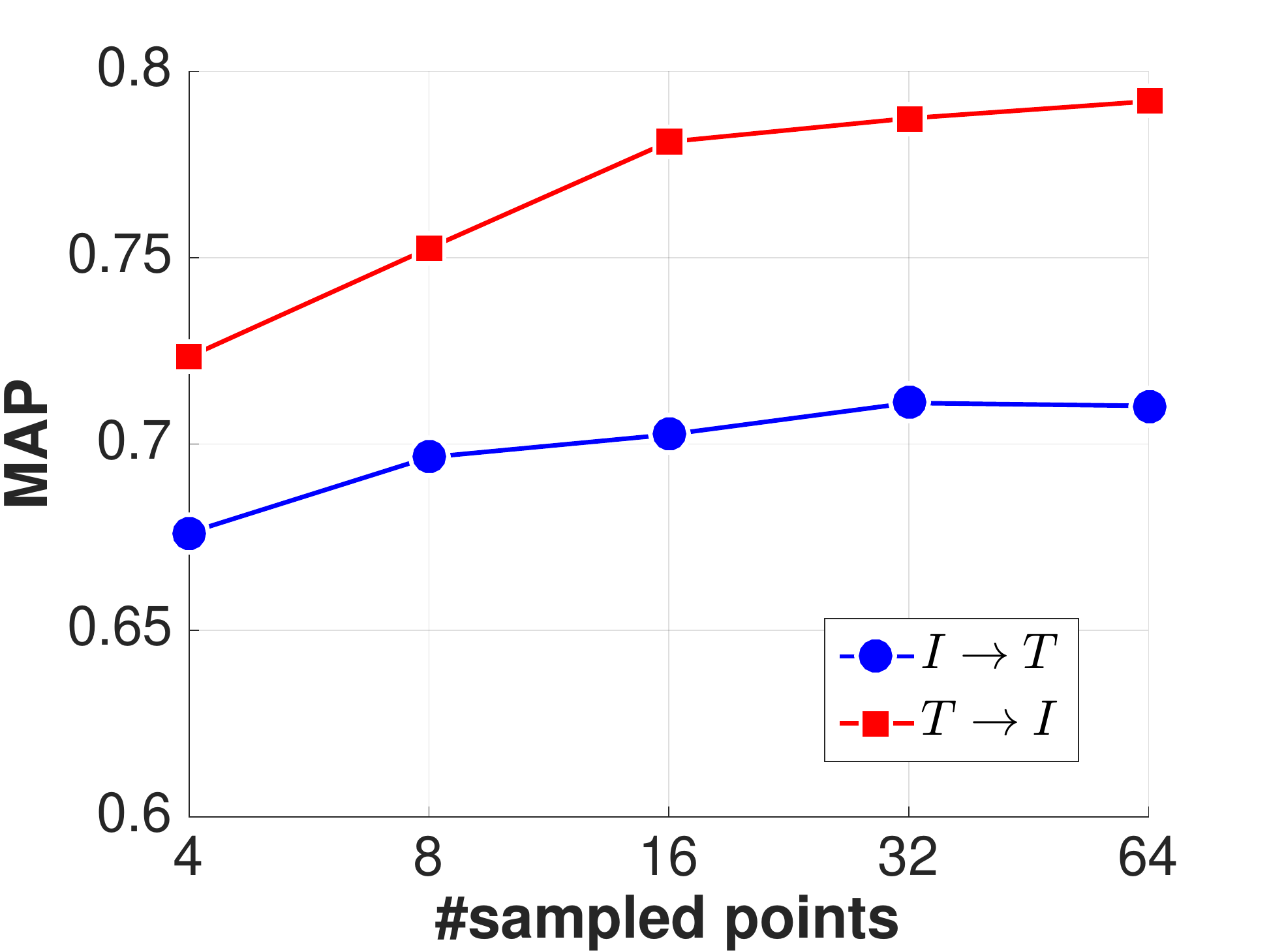}\\
    (c) MIRFLICKR-25K
\end{minipage} &
\begin{minipage}{0.24\linewidth}\centering
    \includegraphics[width=1\textwidth,height = 0.85\textwidth]{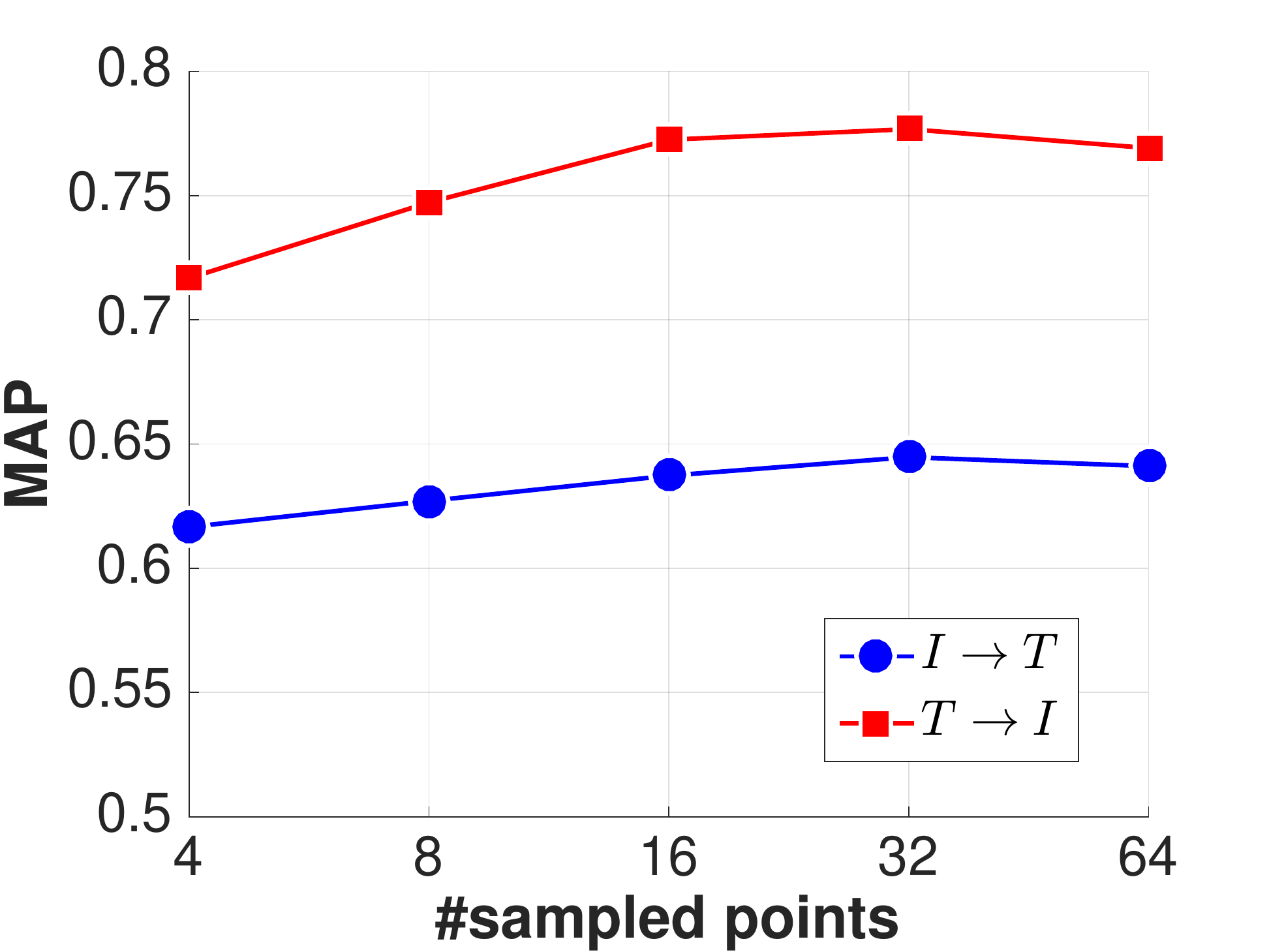}\\
    (d) NUS-WIDE
\end{minipage}
\end{tabular}
\caption{MAP values with different $\lambda$ and different number of sampled points.}
\label{fig:hyperparam}
\end{figure}

\section{Conclusion}\label{conclusion}
In this paper, we propose a novel cross-modal hashing method, called \underline{d}iscrete \underline{l}atent \underline{f}actor model based cross-modal \underline{h}ashing~(DLFH), for cross-modal similarity search in large-scale datasets. DLFH is a discrete method which can directly learn the binary hash codes, and at the same time it is efficient. Experiments on real datasets show that DLFH can significantly outperform relaxation-based continuous methods in terms of accuracy, but with a comparable training speed. Furthermore, DLFH can significantly outperform existing discrete methods in terms of both accuracy and training speed.

\small
\bibliography{example_paper}
\bibliographystyle{abbrv}

\end{document}

%% file: defs.tex
\def\0{{\bf 0}}

\def\A{{\bf A}}

\def\B{{\bf B}}

\def\C{{\bf C}}

\def\H{{\bf H}}
\def\I{{\bf I}}

\def\I{{\bf I}}

\def\X{{\bf X}}
\def\Y{{\bf Y}}

\def\S{{\bf S}}
\def\x{{\bf x}}
\def\y{{\bf y}}

\def\U{{\bf U}}
\def\u{{\bf u}}
\def\VV{{\bf V}}

\def\W{{\bf W}}

\def\X{{\bf X}}

\def\0{{\bf 0}}
\def\1{{\bf 1}}

\def\OM{{\mathcal O}}

\def\RB{{\mathbb R}}

\newtheorem{theorem}{Theorem}